\documentclass[11pt]{article}

\usepackage{verbatim,url,enumerate,color,paralist,cite}
\usepackage{amsmath,amsfonts,nicefrac}
\usepackage{epsfig,amssymb,amstext,xspace}
\usepackage{algorithm}
\usepackage{algorithmicx,algpseudocode}
\usepackage{fullpage}

\newcommand{\qed}{\hspace*{\fill}$\Box$}
\newtheorem{theorem}{Theorem}[section]
\newtheorem{lemma}[theorem]{Lemma}
\newtheorem{corollary}[theorem]{Corollary}

\newtheorem{claim}[theorem]{Claim}

\newenvironment{proof}[1][Proof. ]{\noindent {\bf #1 }}{\qed}
\newenvironment{proofof}[1]{\medskip \noindent {\bf{Proof of #1. }}}{\qed}

\newcommand{\argmax}{{\rm argmax}}

\begin{document}

\title{\Large Asymmetric Traveling Salesman Path and\\ Directed Latency Problems\footnote{A preliminary 
version of this paper appeared in the Proceedings of 21st Annual ACM-SIAM Symposium on 
Discrete Algorithms}}
\author{
Zachary Friggstad\thanks{Supported by NSERC and iCORE scholarships} \hspace{12mm}
Mohammad R. Salavatipour\thanks{Supported by NSERC and an Alberta Ingenuity New Faculty award} \hspace{12mm}
Zoya Svitkina\thanks{Supported by Alberta Ingenuity}\\
\\
\small{Department of Computing Science}\\ 
\small{University of Alberta}\\
\small{Edmonton, Alberta T6G 2E8, Canada}
}

\date{June 1, 2010}
\maketitle

\begin{abstract} 

We study integrality gaps and approximability of two closely related problems on directed graphs. 
Given a set $V$ of $n$ nodes in an underlying asymmetric metric and two
specified nodes $s$ and $t$, both problems ask to find an $s$-$t$ path visiting all other nodes. In 
the {\em asymmetric traveling salesman path problem} (ATSPP), the objective is to minimize the total cost of this path. In the {\em directed latency problem}, the objective is to minimize the sum of distances on this path from $s$ to each node.  Both of these problems are NP-hard. 
The best known approximation algorithms for ATSPP had 
ratio $O(\log n)$ \cite{chekuri:pal:atspp,feige:singh} until the very recent result that improves it to $O(\log n/\log \log n)$ \cite{asadpour:atsp,feige:singh}. However, only a bound of $O(\sqrt{n})$ for the integrality gap of its linear programming relaxation has been known. For directed latency, the best previously known approximation algorithm has a guarantee 
of $O(n^{1/2+\epsilon})$, for any constant $\epsilon>0$ \cite{nagarajan:ravi:latency}.

We present a new algorithm for the ATSPP problem that has an approximation ratio of $O(\log n)$, 
but whose analysis also bounds the integrality gap of the standard LP relaxation of ATSPP by the same factor. This solves an open problem posed in \cite{chekuri:pal:atspp}.  We then pursue
a deeper study of this linear program and its variations, which leads to an algorithm for the $k$-person ATSPP (where $k$ $s$-$t$ paths of minimum total length are sought) and an $O(\log n)$-approximation for the directed latency problem.  

\end{abstract}

\section{Introduction}
Let $G=(V,E)$ be a complete directed graph on a set of $n$ nodes and let 
$d : E \rightarrow \mathbb R^+$ be a cost function
satisfying the directed triangle inequality $d_{uw} \leq d_{uv} + d_{vw}$ for
all $u,v,w \in V$.  However, $d$ is not necessarily symmetric: it may be that
$d_{uv} \neq d_{vu}$ for some nodes $u,v \in V$.
In the metric Asymmetric Traveling Salesman Path Problem (ATSPP), we are also given
two distinct nodes $s,t \in V$.  The goal is to find a path
${s = v_1}, v_2, \ldots, {v_n = t}$
that visits all the nodes in $V$ while minimizing the sum $\sum_{j=1}^{n-1} d_{v_j v_{j+1}}$.
ATSPP can be used to model scenarios such as minimizing the total cost of travel
for a person trying to visit a set of cities on the way from a starting point to a destination. 
This is a variant of the classical Asymmetric Traveling Salesman Problem (ATSP), where the goal is to find a minimum-cost \emph{cycle} visiting all nodes.
In the $k$-person ATSPP, given an integer $k\geq 1$, the goal is to find $k$ paths from $s$ to $t$ such that every node is contained in at least one path and the sum of path lengths is minimized.

Related to ATSPP is the directed latency problem.  
On the same input, the goal is to find a 
path ${s = v_1}, v_2, \ldots, {v_n = t}$ that minimizes the sum of latencies of the nodes.  Here, the latency of node $v_i$ in the path
is defined as $\sum_{j = 1}^{i-1} d_{v_jv_{j+1}}$.
The objective can be thought of as minimizing the total waiting time of clients or
the average response time.
There are possible variations in the problem definition, such as asking for a cycle instead of a path, or specifying only $s$ but not $t$, but they easily reduce to the version that we consider. 
Other names used in the literature for this problem are the {\em deliveryman problem}
\cite{minieka} and the {\em traveling repairman problem} \cite{afrati}.

\subsection{Related work}~
Both ATSPP and the directed latency problem are closely
related to the classical Traveling Salesman Problem (TSP), 
which asks to find the cheapest Hamiltonian cycle in a complete undirected graph with edge costs \cite{lawler:lenstra,gutin}.
In general weighted graphs, TSP is not approximable. However, in most practical settings it can be assumed
that edge costs satisfy the triangle inequality
(i.e.\ $d_{uw} \leq d_{uv} + d_{vw}$).
Though metric TSP is still NP-hard, the well-known algorithm of 
Christofides \cite{christofides} has an approximation ratio of \nicefrac{3}{2}.
Later the analysis in \cite{shmoys:williamson,wolsey:heuristic} showed that this approximation algorithm actually bounds the integrality
gap of a linear programming relaxation for TSP known as the Held-Karp LP.
This integrality gap is also known to be 
at least $\nicefrac{4}{3}$.  Furthermore, for all
$\epsilon > 0$, approximating TSP within a factor of $\nicefrac{220}{219}-\epsilon$ is NP-hard \cite{papadimitriou:vempala}.
Christofides' heuristic was adapted to the problem of finding the cheapest Hamiltonian path
in a metric graph with an approximation guarantee of \nicefrac{3}{2} if at most one endpoint is specified or \nicefrac{5}{3} if both endpoints are given \cite{hoogeveen}.

In contrast to TSP, no constant-factor approximation for its asymmetric version is known. The current best approximation for ATSP is the very recent result of Asadpour et al.\ \cite{asadpour:atsp}, which gives an $O(\log n/ \log \log n)$-approximation algorithm. It also upper-bounds the integrality gap of the asymmetric Held-Karp LP relaxation by the same factor. 
Previous algorithms guarantee a solution of cost within  $O(\log n)$ factor 
of optimum \cite{frieze:galbiati,kleinberg:williamson,kaplan,feige:singh}. The 
algorithm of Frieze et al.\ \cite{frieze:galbiati} is shown to
upper-bound the  Held-Karp integrality gap by $\log_2n$ in \cite{williamson:msthesis}, and a different proof that bounds the integrality gap of a slightly weaker LP is obtained in 
\cite{nagarajan:ravi:polylog}. The best known lower bound on the Held-Karp integrality gap is essentially 2 \cite{charikar:goemans:karloff}, and tightening these bounds remains an important open problem. ATSP is NP-hard to approximate 
within $\nicefrac{117}{116} - \epsilon$ \cite{papadimitriou:vempala}.

The path version of the problem, ATSPP, has been studied much less than ATSP, but there are some recent results concerning its approximability.
An $O(\sqrt{n})$ approximation algorithm for it was given by Lam and 
Newman \cite{lam:newman}, which 
was subsequently improved to $O(\log n)$ by Chekuri and Pal \cite{chekuri:pal:atspp}. 
Feige and Singh \cite{feige:singh} improved upon this guarantee by a constant factor and also showed that the 
approximability of ATSP and \mbox{ATSPP} are within a constant factor of each other, i.e.\ an $\alpha$-approximation for one implies an $O(\alpha)$-approximation for the other. Combined with the result of \cite{asadpour:atsp}, this implies an $O(\log n/ \log \log n)$ approximation for ATSPP. 
However, none of these algorithms bound the integrality gap of the LP relaxation for ATSPP. This integrality gap was considered by Nagarajan and 
Ravi \cite{nagarajan:ravi:latency}, who showed that it is at most $O(\sqrt{n})$.
To the best of our knowledge, the asymmetric path version of the $k$-person problem has not been studied previously. However, some work has been done on its symmetric version, where the goal is to find $k$ rooted cycles of minimum total cost (e.g., \cite{frieze:kperson}).

The metric minimum latency problem is NP-hard for both the undirected and directed versions since an exact algorithm for either of these could be used to efficiently solve the Hamiltonian Path problem.
The first constant-factor approximation for minimum latency  on undirected graphs was 
developed by Blum et al.\ \cite{blum:chalasani}. This was subsequently improved in 
a series of papers from 144 to 21.55 \cite{goemans:kleinberg}, then to 7.18 \cite{archer:levin:williamson} and ultimately to 3.59 \cite{chaudhuri}. 
Blum et al.\ \cite{blum:chalasani} also observed that
there is some constant $c$ such that there
is no $c$-approximation for  minimum latency unless P = NP.
For directed graphs, Nagarajan and Ravi \cite{nagarajan:ravi:latency}
gave an $O((\rho+\log n)\, n^{\epsilon}\, \epsilon^{-3})$ approximation algorithm that runs in time $n^{O(1/\epsilon)}$,
where $\rho$ is the integrality gap of an LP relaxation for ATSPP.
Using their $O(\sqrt n)$ upper bound on $\rho$, they 
obtained a guarantee of $O(n^{1/2 + \epsilon})$, which is the best approximation ratio known for this problem before our present results.

\subsection{Our results}~
In this paper we study both the ATSPP and the directed latency problem.
\begin{figure}[ht]
\begin{align} 
\mbox{min}~~ & \sum_{e\in E} d_e x_e  \label{alphalp} \\
\mbox{s.t.}~~ & x(\delta^+(u))=x(\delta^-(u)) \hspace{9mm} \forall u\in V \setminus \{s,t\} \label{flcon} \\
& x(\delta^+(s))=x(\delta^-(t))=1  \label{con2}  \\
& x(\delta^-(s))=x(\delta^+(t))=0  \label{con3} \\
& x(\delta^-(S))\geq \alpha \hspace{20mm} \forall S\subset V, S \neq \emptyset, s\notin S \label{alp-sc} \\
& x_e \geq 0  \hspace{32mm}  \forall e\in E  \notag 
\end{align}
\end{figure}
The natural LP relaxation for ATSPP is (\ref{alphalp}) with $\alpha=1$, where $\delta^+(\cdot)$ denotes the set of outgoing edges from a vertex or a set of vertices, and $\delta^-(\cdot)$ denotes the set of incoming edges. A variable $x_e$ indicates that edge $e$ is included in a solution.
Let us refer to this linear program as LP($\alpha$), 
as we study it for different values of $\alpha$ in constraints~(\ref{alp-sc}). 
We begin in Section \ref{sec:atspp-ig} by proving that the integrality gap of LP($\alpha=1$) is $O(\log n)$.

\begin{theorem}\label{thm:ig}
If $L$ is the cost of a feasible solution to LP (\ref{alphalp})
with $\alpha =1$, then one can find, in polynomial time,
a Hamiltonian path from $s$ to $t$ with cost at most
$(2 \log n +1) \cdot L$.\footnote{All logarithms in this paper are base 2.}
\end{theorem}

We note that, despite bounding the integrality gap, our algorithm is actually combinatorial and does not require solving the LP. 
We strengthen the result of Theorem \ref{thm:ig} by extending it to any 
$\alpha$ with $\frac{1}{2} < \alpha \leq 1$. 
This captures the LP of \cite{nagarajan:ravi:latency}, which has $\alpha=\frac{2}{3}$, and is also used in our algorithm for the directed latency problem.
We prove the following theorem in Section \ref{sec:atspp-alpha1}.

\begin{theorem}\label{thm:a-ig}
If $L$ is the cost of a feasible solution to LP (\ref{alphalp})
with $\frac{1}{2} < \alpha \leq 1$, then one can find, in polynomial time,
a Hamiltonian path from $s$ to $t$ with  cost at most $\frac{6\log n +3}{2\alpha-1} \cdot L$.
\end{theorem}

It is worth observing that this theorem, together with the results of \cite{nagarajan:ravi:latency}, imply a polylogarithmic approximation algorithm for the directed latency problem which runs in quasi-polynomial time, as well as a polynomial-time $O(n^\epsilon)$-approximation.
However, that approach relies on guessing a large number of intermediate vertices of the path, and thus does not yield an algorithm that has both a polynomial running time and a polylogarithmic approximation guarantee. 
So, to obtain a polynomial-time approximation, we use a different approach. 
For that we consider LP($\alpha$) for values of $\alpha$ that include $\alpha\leq \frac{1}{2}$.
If we allow $\alpha \leq \frac{1}{2}$ then LP($\alpha$),
as a relaxation of ATSPP, can be shown to have an unbounded integrality gap.
However, we prove the following theorem in Section \ref{sec:atspp-alpha2}.

\begin{theorem}\label{thm:kpath}
If $L$ is the cost of a feasible solution to 
LP (\ref{alphalp}) with
$\alpha = \frac{1}{k}$, for integer $k$, then one can find, in polynomial time,
a collection of at most $k\cdot \log n$ paths from $s$ to $t$, such 
that each vertex of $G$ appears on at least one path,
and the total cost of all these paths is at most $kL \log n$.
\end{theorem}

Next, we study another generalization of the ATSPP, namely the $k$-person asymmetric traveling salesman path problem. In Section \ref{sec:kperson} we prove the following theorem:

\begin{theorem}\label{thm:kperson}
There is an $O(k^2 \log n)$ approximation algorithm for the $k$-person ATSPP. 
Moreover, the integrality gap of its LP relaxation is bounded by the same factor.
\end{theorem}

Given these results concerning LP($\alpha$), we study a particular
LP relaxation for the directed latency problem in Section \ref{sec:latency}.
We improve upon the $O(n^{1/2+\epsilon})$-approximation of \cite{nagarajan:ravi:latency} substantially by proving the following:

\begin{theorem} \label{thm:lat}
A solution to the directed latency problem can be found in polynomial time
that has cost no more than $O(\log n)\cdot L$, where $L$ is the value 
of LP relaxation (\ref{lp:lat}), which is also a lower bound on the integer optimum.
\end{theorem}

We note that this seems to be the first time that 
a bound is placed on the integrality gap of any LP relaxation for the minimum latency problem, even in the undirected case.

\section{Integrality gap of ATSPP} \label{sec:atspp-ig}

We show that LP relaxation (\ref{alphalp}) of ATSPP with $\alpha=1$ has integrality gap of $O(\log n)$.
Let $x^*$ be its optimal fractional solution, and let $L$ be its cost.
We define a \emph{path-cycle cover} on a subset of vertices $W\subseteq V$ containing $s$ and $t$ to be the union of one $s$-$t$ path and zero or more cycles, such that each $v\in W$ occurs in exactly one of these subgraphs. The cost of a path-cycle cover is the sum of costs of its edges.

Our approach is an extension of the algorithm by Frieze et al.\,\cite{frieze:galbiati}, analyzed by Williamson\,\cite{williamson:msthesis} to bound the integrality gap for ATSP. That algorithm finds a minimum-cost cycle cover on the current set of vertices, chooses an arbitrary representative vertex for each cycle, deletes other vertices of the cycles, and repeats, at the end combining all the cycle covers into a Hamiltonian cycle. As this is repeated at most $\log n$ times, and the cost of each cycle cover is at most the cost of the LP solution, the upper bound of $\log n$ on the integrality gap is obtained.  In our algorithm for ATSPP, the analogue of a cycle cover is a path-cycle cover (also used in \cite{lam:newman}), whose cost is at most the cost of the LP solution (Lemma \ref{lem:pathcyc}). At the end we combine the edges of $O(\log n)$ path-cycle covers to produce a Hamiltonian path. However, the whole procedure is more involved than in the case of ATSP cycle. For example, we don't choose arbitrary representative vertices, but use an amortized analysis to ensure that each vertex only serves as a representative a bounded number of times.

We note that a path-cycle cover of minimum cost can be found by a combinatorial algorithm, using a reduction to minimum-cost perfect matching, as explained in \cite{lam:newman}.  
In the proof of Lemma \ref{lem:pathcyc} below, we make use of the following splitting-off theorem, as also done in \cite{nagarajan:ravi:polylog}, where splitting off edges $yv$ and $vx$ refers to replacing these edges with the edge $yx$ (unless $y=x$, in which case the two edges are just deleted). 

\begin{theorem}[Frank \cite{frank} and Jackson \cite{jackson}] \label{thm:so}
Let $G=(V,E)$ be a Eulerian directed graph and $vx\in E$. There exists an edge $yv\in E$ such that splitting off $yv$ and $vx$ 
does not reduce the directed connectivity from $u$ to $w$ for any $u,w \in V\setminus \{v\}$.
\end{theorem}
This theorem also applies to weighted Eulerian graphs, i.e.\ ones in which the weighted out-degree of every vertex is equal to its weighted in-degree, since weighted edges can be replaced by multiple parallel edges, producing an unweighted Eulerian multigraph.

\begin{lemma} \label{lem:pathcyc}
For any subset $W\subseteq V$ that includes $s$ and $t$, there is 
a path-cycle cover of $W$ of cost at most $L$.
\end{lemma}
\begin{proof}
Consider the graph $G'$ obtained from $G$ by assigning capacities $x^*_e$ to edges $e\in E$ and adding a dummy edge from $t$ to $s$ with unit capacity. From constraints (\ref{flcon})-(\ref{con3}) of LP($\alpha=1$) it follows that this is a weighted Eulerian graph. Constraints (\ref{alp-sc}) and the max-flow min-cut theorem imply that for any $v\in V$, the directed connectivity from $s$ to $v$ in $G'$ is at least $\alpha=1$.

We apply the splitting-off operation on $G'$, as guaranteed by Theorem \ref{thm:so}, to vertices in $V\setminus W$ until all of them are disconnected from the rest of the graph. Let $G''$ be the resulting graph on $W$ and let $x'$ be its edge capacities except for the dummy edge $ts$ (which was unaffected by the splitting-off process).
By Theorem \ref{thm:so}, the directed connectivity from $s$ to any $v\in W$ does not decrease from the splitting-off operations, which means that in $G''$ it is still at least $\alpha$. This ensures that $x'$ satisfies constraints (\ref{alp-sc}) for all sets $S\subset W$ with $S\neq \emptyset$ and 
$s\notin S$, and is a feasible solution to LP($\alpha$) on the subset $W$ of vertices. Furthermore, the triangle inequality implies that the cost of $x'$ is no more than that of $x^*$, namely $L$.

Now we make the observation that if we remove from LP($\alpha=1$) constraints (\ref{alp-sc}) for all but singleton sets, the resulting LP is equivalent to a circulation problem, and thus has an integer 
optimal solution. Since there is a feasible solution to LP($\alpha=1$) on the set $W$ of cost at most $L$ (namely $x'$), and removing a 
constraint can only decrease the optimal objective value, it means that there is an integer solution to the following program that costs no more than $L$:
\begin{align}
\mbox{min}~~ & \sum_{e\in E} d_e x_e \label{intlp} \\
\mbox{s.t.}~~ & x(\delta^+(u))=x(\delta^-(u)) \geq 1 & \forall u\in W \setminus \{s,t\} \notag \\
& x(\delta^+(s))=x(\delta^-(t))=1 \label{amtfl} \\
& x(\delta^-(s))=x(\delta^+(t))=0  \notag \\
& x_e \geq 0 & \forall e\in E  \notag 
\end{align}
In principle, this integer solution can have $x(\delta^+(u))>1$ 
for some nodes $u$. In this case, we find a Euler tour of each component of the resulting graph (with dummy edge $ts$ added in) and shortcut it over any repeated vertices. This ensures that $x(\delta^+(u))=1$ for all $u$ without increasing the cost. 
But such a solution is precisely a path-cycle cover of $W$.
\end{proof}

\medskip

\begin{algorithm*}[ht]
  \caption{~Asymmetric Traveling Salesman Path} \label{alg:ig} 
\begin{algorithmic}[1] 
\State Let a set $W \leftarrow V$; integer labels $l_v \leftarrow 0$ for all $v\in V$; flow $F \leftarrow \emptyset$ and circulation $H \leftarrow \emptyset$  \label{line:one}
\For {$2 \log_2 n +1$ iterations} \label{line:iter}
\State Find the minimum-cost $s$-$t$ path-cycle cover $F'$ on $W$ 
\label{line:findflow}
\State $F\leftarrow F+F'$ \label{line:addf}
\Comment $F$ is acyclic before this operation
\State Find a path-cycle decomposition of $F$, with cycles $C_1...C_k$ and paths $P_1...P_h$, such that $\bigcup_i P_i$ is acyclic
\For {each connected component $A$ of $\bigcup_j C_j$} \label{line:finda}
\Comment $A$ is a circulation
\State For each vertex $u\in A$, let $d_u$ be the in-degree of $u$ in $A$
\State Find a ``representative'' node $v\in A$ minimizing $l_v+d_v$ \label{line:findv}
\State $F \leftarrow F - A$
\Comment subtract flows
\State {\bf for} each $w\in A$, $w\neq v$, and for each path $P_i$ 
\State ~~~~~{\bf if} $w\in P_i$ {\bf then} modify $F$ by shortcutting $P_i$ over $w$
\State Remove all nodes in $A$, except $v$, from $W$ \label{line:delw}
\Comment Note: they don't participate in $F$ anymore
\State $H\leftarrow H+A$
\Comment add circulations
\State $l_v \leftarrow l_v+d_v$
\EndFor
\EndFor \label{line:endfor}
\State Let $P$ be an $s$-$t$ path consisting of nodes in $W$ in the order found by  	topologically sorting $F$ \label{line:findp} \newline \mbox{}
\Comment $F$ is an acyclic flow on the nodes $W$
\For {every connected component $X$ of $H$ of size $|X|>1$} \label{line:cyc}
\State Find a Euler tour of $X$, shortcut over nodes that appear more than once \label{line:euler}
\State Incorporate the resulting cycle into $P$ using a shared node \label{line:inc}
\EndFor \label{line:endcyc}
\State {\bf return} $P$
\Comment $P$ is a Hamiltonian $s$-$t$ path
\end{algorithmic}
\end{algorithm*}

We consider Algorithm~\ref{alg:ig}. Roughly speaking, the idea is to find a path-cycle cover, select a representative node for each cycle, delete the other cycle nodes, and repeat. Actually, a representative is selected for a component more general than a simple cycle, namely a union of one or more cycles. We ensure that each vertex is selected as a representative at most $\log n$ times, which means that after $2 \log n +1$ iterations, each surviving vertex has participated in the acyclic part of the path-cycle covers at least $\log n +1$ times. This allows us at the end to find an $s$-$t$ path which spans all the surviving vertices, $W$, and consists entirely of edges in the acyclic part, $F$, of the union of all the path-cycle covers, using a technique of \cite{nagarajan:ravi:latency}. Then we insert into it the subpaths obtained from the cyclic part, $H$, of the union of path-cycle covers, connected through their representative vertices. 
We occasionally treat subgraphs satisfying appropriate degree constraints as flows or circulations.

\begin{lemma} \label{lem:label}
During the course of the algorithm, no label $l_v$ exceeds the value $\log n$.
\end{lemma}

The idea of the proof is, as the algorithm proceeds, to maintain a forest on the set of nodes $V$, such that the number of leaves in a subtree rooted at any node $v\in V$ is at least $2^{l_v}$. The lemma then follows because the total number of leaves is at most $n$. 
We first prove an auxiliary claim.

\begin{claim} \label{claim:xy}
In each component $A$ found by  Algorithm \ref{alg:ig} on line \ref{line:finda}, there are two distinct nodes $x$ and $y$ such that $d_x = d_y = 1$.
\end{claim}
\begin{proof}
Let $\bar{F}$ be the value of $F$ at the start of the current iteration of the outside loop, 
i.e. before $F'$ is added to it on line \ref{line:addf}. $\bar{F}$ is acyclic, because during 
the course of the loop, all cycles of $F$ are subtracted from it. So $A$ is a union of cycles, 
formed from the sum of an acyclic flow $\bar{F}$ and a path-cycle cover $F'$, which sends exactly 
one unit of flow through each vertex. 

Consider a topological ordering of nodes based on the flow $\bar{F}$, and let $x$ and $y$ be 
the first and last nodes of $A$, respectively, in this ordering. As $A$ always contains at least 
two nodes, $x$ and $y$ are distinct. Since $x$ and $y$ participate in some cycle(s) in $A$, their 
in-degrees are at least 1. We now claim that the in-degree of $x$ in $A$ is at most 1.  Indeed, 
since all other nodes of $A$ are later than $x$ in the topological ordering, it cannot have any 
flow coming from them in $\bar{F}$. So the only incoming flow to $x$ can be in $F'$.  But since 
$F'$ sends a flow of exactly one unit through each vertex, the in-degree of $x$ in $A$ is at most 
one.  A symmetrical argument can be made for $y$, showing that its out-degree in $A$ is at most 
one.  But since $A$ is a union of cycles, every node's in-degree is equal to its out-degree, and 
the in-degree of $y$ is also at most 1.
\end{proof}

\begin{proofof}{Lemma \ref{lem:label}}
As the algorithm proceeds, let us construct a forest on the set of nodes $V$. Initially, each 
node is the root of its own tree. We maintain the invariant that $W$ is the set of tree roots 
in this forest. For each component $A$ that the algorithm considers, and the node $v$ found 
on line \ref{line:findv}, we attach the nodes of $A$, except $v$, as children of $v$. Note 
that the invariant is maintained, as these nodes are removed from $W$ on line \ref{line:delw}. 
The set of nodes of each component $A$ found on line \ref{line:finda} is always 
a subset of $W$, and thus our construction indeed produces a forest. 

We show by induction on the steps of the algorithm that if a node has label $l$, then its 
subtree contains at least $2^l$ leaves. Thus, since there are $n$ nodes total, no label can 
exceed $\log_2 n$. At the beginning of the algorithm, all labels are 0, and all trees have 
one leaf each, so the base case holds. Now consider some iteration in which the label of 
vertex $v\in A$ is increased from $l_v$ to $l_v + d_v$. By Claim  \ref{claim:xy}, there are 
nodes $x,y\in A$ (possibly one of them equal to $v$) with $d_x=d_y=1$. Since $v$ minimizes 
$l_u+d_u$ among all vertices  $u\in A$, we have that $l_x + d_x \geq l_v+d_v$ and 
$l_y + d_y \geq l_v+d_v$, and thus $l_x \geq l_v+d_v -1$ and $l_y \geq l_v+d_v -1$. 
Thus, by the induction hypothesis, the trees rooted at $x$ and $y$ each have at least 
$2^{l_v+d_v-1}$ leaves. Because we update the forest in such a way that $v$'s new tree 
contains all the leaves of trees previously rooted at $x$ and $y$, this tree now has at 
least $2\cdot 2^{l_v+d_v-1} = 2^{l_v+d_v}$ leaves. 
\end{proofof}

\begin{lemma} \label{lem:amtflow}
At the end of the algorithm's main loop, the flow in $F$ passing through any node $v\in W$ 
is equal to $2\log n + 1 - l_v$, and thus (by Lemma \ref{lem:label}) is at least $\log n +1$.
\end{lemma}
\begin{proof}
There are $2\log n + 1$ iterations, each of which adds one unit of flow through each vertex $v\in W$. 
We now claim that for a vertex $v\in W$, the amount of flow removed from it is equal to its label, 
$l_v$. Flow is removed from $v$ only if $v$ becomes part of some component $A$. Now, if it is ever 
part of $A$, but not chosen as a representative on line $\ref{line:findv}$, then it is removed 
from $W$. Thus, we are only concerned about vertices that are chosen as representatives every 
time that  they are  part of $A$. Such a vertex has flow $d_v$ going through it in $A$, which 
is the amount subtracted from $F$. But since this is also the amount by which its label increases, 
the lemma follows.
\end{proof}

\medskip

We now show that Algorithm \ref{alg:ig} returns a Hamiltonian $s$-$t$ path of cost at most $(2\log_2 n +1)\cdot L$.

\begin{proofof}{Theorem \ref{thm:ig}}
At the end of the main loop, all nodes of $V$ are part of either $W$ or $H$ or both. So when 
all components of $H$ are incorporated into the path $P$, all nodes of $V$ become part of the path. 
We bound the cost of all the edges used in the final path by the total cost of all the path-cycle covers found on 
line \ref{line:findflow} of the algorithm. We note that at the end of the algorithm, the cost of 
the flow $F+H$ is no more than this total.

We claim that when the $s$-$t$ path $P$ is found on line \ref{line:findp},  $F$ contains flow on 
every edge between consecutive nodes of $P$. 
This is similar to an argument used in \cite{nagarajan:ravi:latency}.
First, since $F$ is acyclic, it has a topological ordering. Suppose we find a flow decomposition of $F$ into paths. There are at most $2\log n +1$ 
such paths, and, by Lemma \ref{lem:amtflow}, each vertex of $W$ participates in at least 
$\log n +1$, or more than half, of them.  This means that any two vertices $u,v\in W$ must 
share a path, say $P'$, in this decomposition. In particular, suppose that $v$ immediately 
follows $u$ in the path $P$. This means that $v$ appears later than $u$ in the topological 
order, so on $P'$ $v$ comes after $u$. Moreover, we claim that on $P'$, $v$ will be the immediate 
successor of $u$. If not, suppose that there is a node $w$ that appears between $u$ and $v$ in $P'$. 
But this means that in the topological ordering (and thus in $P$), $w$ will appear after $u$ and 
before $v$, which contradicts the fact that they are consecutive in $P$. So we conclude that there 
is an edge with flow in $F$ between any two consecutive nodes of $P$, and thus the path $P$ costs 
no more than the flow $F$. 

Regarding $H$, we note that it is a sum of cycles, and thus Eulerian. So it is possible to find a 
Euler tour of each of its components, using only edges with flow in $H$. The subsequent shortcutting 
can only decrease the cost. Thus, the total cost of cycles found on line \ref{line:euler} is no more 
than the cost of the flow $H$. To describe how these cycles are incorporated into the path $P$, we 
show that each of them (or, equivalently, each connected component of $H$) shares exactly one node with $W$ (and thus with $P$). Note that every  component $A$ added to $H$ contains only 
nodes that are in $W$ at that time. Moreover, when this is done, all but one nodes of $A$ are 
expelled from $W$. So when several components of $H$ are connected by the addition of $A$, the 
invariant is maintained that there is one node per component that is shared with $W$. Now, 
suppose that $v$ is the vertex shared by the cycle obtained from component $X$ and the path $P$. 
On line \ref{line:inc}, we incorporate the cycle into the path by following the path up to $v$, 
then following the cycle up to the predecessor of $v$, then connecting it to the successor of $v$ 
on the path. By triangle inequality, the resulting longer path costs no more than the sum of costs 
of the old path and the cycle.
\end{proofof}

\section{Integrality gap for relaxed ATSPP LP} \label{sec:atspp-alpha1}

Consider LP($\alpha$) 
with $\frac{1}{2} < \alpha \leq 1$, and say that it has cost $L$. We bound its integrality gap for ATSPP. 
As in the proof of Lemma \ref{lem:pathcyc}, we can apply splitting-off to obtain a feasible solution 
to LP($\alpha$), of cost at most $L$, on a graph induced by a subset of vertices $W\subseteq V$. 
Let $x$ be such a solution. Lemma \ref{lem:tr} below shows how to use $x$ to find a feasible fractional 
solution to LP (\ref{intlp}) on $W$, of cost within a constant factor of $L$, namely $\frac{3}{2\alpha-1} L$. 
Since LP (\ref{intlp}) has integer optimum, there is an integer solution to  LP (\ref{intlp}) on $W$, and thus a path-cycle cover,  
of cost at most $\frac{3}{2\alpha-1} L$. Then we can proceed as in Section \ref{sec:atspp-ig}, applying Algorithm \ref{alg:ig} to bound the cost of the resulting ATSPP solution by $2\log n +1$ times the path-cycle cover cost. This shows that LP($\alpha$) has integrality gap at most 
$\frac{6\log n +3}{2\alpha-1}$, proving Theorem \ref{thm:a-ig}.

\begin{lemma} \label{lem:tr}
Given a solution $x$ to LP($\alpha$), with $\alpha>1/2$, on a subset $W\subseteq V$, with cost at most $L$, a feasible solution to LP (\ref{intlp}) on $W$ of cost at most $\frac{3}{2\alpha-1} L$ can be found.
\end{lemma}
\begin{proof}
Multiply $x$ by $1/\alpha$. Now it constitutes a flow $F$ of $1/\alpha$ units from $s$ to $t$. 
Constraints (\ref{alp-sc}), restricted to sets of size 1, imply that each node $u$ now has at 
least one unit of flow going through it. Find a flow decomposition of $F$ into paths and cycles, 
so that the union of the paths is acyclic. Let $F=F_p+F_c$, where $F_p$ is the sum of flows on 
the paths in our decomposition, and $F_c$ is the sum of flows on the cycles. 

Choose some $\gamma$ such that $\frac{1}{2\alpha}<\gamma < 1$. For any node $u$ such that the 
amount of $F_p$ flow going through $u$ is less than $\gamma$, shortcut any flow decomposition 
paths that contain $u$, so that there is no more $F_p$ flow going through $u$. Let $U\subseteq W$ 
be the set of vertices still participating in the $F_p$ flow. Then each vertex in $U$ has at least 
$\gamma$ units of $F_p$ flow going through it, and each vertex in $W\setminus U$ has at least 
$1-\gamma$ units of $F_c$ flow going through it.  

We find a topological ordering of vertices in $U$ according to $F_p$ (which is acyclic), and 
let $P$ be an $s$-$t$ path that visits the nodes of $U$ in this topological order. We claim 
that the cost of $P$ is within a constant factor of the cost of $F_p$.
The argument for this is similar to one in the proof of Theorem \ref{thm:ig}.
Out of $1/\alpha$ units of flow 
going from $s$ to $t$ in $F_p$, each vertex $u\in U$ carries $\gamma$ units, which is more than half of the total amount 
(as $\gamma>1/2\alpha$). So for any two such vertices $u$ and $v$, there must be shared flow paths that carry flow of at least $2\gamma-1/\alpha$ units. 
In particular, for every 
two consecutive nodes $u,v\in P$, $F_p$ must contain such shared paths in which $v$ immediately follows $u$.  So the cost of $P$ is at most $\frac{1}{2\gamma-1/\alpha}$ times the cost of $F_p$.

We now define $\tilde{x}$ as a flow equal to one unit of $s$-$t$ flow on the path $P$ plus 
$\frac{1}{1-\gamma}$ times the flow $F_c$.  We claim that $\tilde{x}$ is a feasible solution 
to LP (\ref{intlp}): there is exactly one unit of flow from $s$ to $t$ (as $F_c$ consists of 
cycles not containing $s$ or $t$); there is flow conservation at all nodes except $s$ and $t$; 
each vertex in $U$ (and thus in $P$) has at least one unit of flow going through it; and each 
vertex in $W\setminus U$ has at least one unit of flow going through it (as it had at least 
$1-\gamma$ units of $F_c$ flow).  The cost of this solution is at most 
\begin{eqnarray*}
 \frac{1}{2\gamma-1/\alpha}\cdot cost(F_p) + \frac{1}{1-\gamma}\cdot cost(F_c)
 & \leq&  \max\left(\frac{1}{2\gamma-1/\alpha}, \frac{1}{1-\gamma}\right) 
\cdot \frac{1}{\alpha}\, L.
\end{eqnarray*}
If we set $\gamma = \frac{1}{3} + \frac{1}{3\alpha}$, which satisfies $\frac{1}{2\alpha}<\gamma < 1$, 
we see that the cost of $\tilde{x}$ is at most $\frac{3}{2\alpha-1}\cdot L$.
\end{proof}

\section{Relaxed ATSPP LP with $\alpha \leq 1/2$} \label{sec:atspp-alpha2}

Consider LP (\ref{alphalp}) with $\alpha = \frac{1}{k} \leq \frac{1}{2}$ for some integer $k \geq 2$. 
It can be shown that, as a relaxation for the ATSPP problem, this LP has unbounded integrality gap.
For example, let $D$ be an arbitrarily large value and consider the shortest path metric obtained from
the graph in Figure \ref{fig:badgap}. One can verify that the following assignment of $x$-values
to the arcs is feasible for LP (\ref{alphalp}) with $\alpha = \nicefrac{1}{2}$.
Assign a value of \nicefrac{1}{2} to arcs $(1,2)$, $(3,2)$, $(3,6)$, $(1,4)$, $(5,4)$, 
and $(5,6)$ and a value of 1 to arcs $(2,3)$ and $(4,5)$.
Every other arc is assigned a value of 0. This assignment is feasible for the linear program and
has objective function value 5. On the other hand, any Hamiltonian path
from 1 to 6 has cost at least $D$.

\begin{figure}
\begin{center}
\includegraphics[scale=0.45]{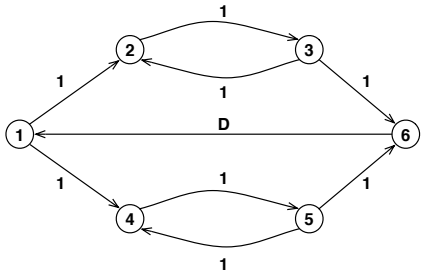}
\caption{Bad gap example for LP (\ref{alphalp}) with $\alpha = \nicefrac{1}{2}$. Here, $D$ is an arbitrarily large integer.}\label{fig:badgap}
\end{center}
\end{figure}

Let $L$ be the cost of the optimal solution to LP($\alpha = \frac{1}{k}$). We show 
how to find $k\cdot \log n$ paths from $s$ to $t$, such that each node of $G$ appears on at least 
one path, and the total cost of all these paths is at most $k \log n \cdot L$. Let us define a 
$k$-path-cycle cover to be a set of $k$ disjoint 
paths from $s$ to $t$ and zero or more cycles, which together cover all nodes. 
Like a path-cycle cover, the minimum-cost $k$-path-cycle cover can be found by a combinatorial algorithm
by creating $k$ copies of both $s$ and $t$ and using the matching algorithm described in \cite{lam:newman}.

\begin{lemma} \label{lem:kpathcyc}
For any subset $W\subseteq V$ that includes $s$ and $t$, there is a 
$k$-path-cycle cover of $W$ with total cost at most $k L$.
\end{lemma}
\begin{proof}
As in the proof of Lemma \ref{lem:pathcyc}, we apply splitting-off to a solution of LP($\alpha$) to get a solution to LP($\alpha$) on a subset $W\subseteq V$ of vertices, of 
no greater cost. Now, if we multiply this solution by $k$, we get a feasible solution to 
LP (\ref{intlp}) with constraints (\ref{amtfl}) replaced with 
$$x(\delta^+(s))=x(\delta^-(t))=k.$$
The cost of this solution is no more than $kL$. 
But this LP has an integer optimum, which, possibly after shortcutting, is exactly a $k$-path-cycle cover. 
\end{proof}

\begin{proofof}{Theorem \ref{thm:kpath}}
We start with $W=V$ and repeat the following until $W=\{s,t\}$: 
\begin{compactenum}
\item
Find a $k$-path-cycle cover $F$ of $W$ of cost at most $kL$, as guaranteed by Lemma \ref{lem:kpathcyc}.
\item
Remove from $W$ all nodes (except $s$ and $t$) that participate in paths of $F$.
\item
For each cycle $C$ of $F$, choose a representative node $v\in C$, and remove from $W$ all nodes of $C$ except $v$.
\end{compactenum}

\noindent
Let us say that the procedure terminates after $T$ iterations. As the size of $W$ halves in each iteration, 
$T$ is at most $\log n$. In the last iteration, all elements of $W$ must have participated in the paths, 
as otherwise there would be a node $v$ that remains in $W$ after this iteration. This implies that the graph 
$\bigcup F$ is connected. It also has total cost at most $kL \log n$, and, if we add $kT$ edges from $t$ to 
$s$, becomes Eulerian. This means that we can construct 
$kT\leq k\log n$ paths from $s$ to $t$, covering all nodes of $V$, out of edges of $\bigcup F$. 
\end{proofof}

\section{Algorithm for $k$-person ATSPP}\label{sec:kperson}
In this section we consider the 
$k$-person asymmetric traveling salesman path problem. 
The LP relaxation for this problem is similar to LP (\ref{alphalp}), but with 
$x(\delta^+(s))=x(\delta^-(t))=k$ for constraint (\ref{con2}) and $x(\delta^-(S))\geq 1$ for constraint (\ref{alp-sc}). 
Arguments similar to those in Section \ref{sec:atspp-ig} show that a $k$-path-cycle cover on any subset $W$ is a 
lower bound on the value of the LP relaxation for the $k$-person ATSPP. Our algorithm constructs a solution that 
uses each edge of $O(k \log n)$ $k$-path-cycle covers at most $k$ times, proving a bound of $O(k^2 \log n)$ on the  approximation ratio and the integrality gap.

Our algorithm starts by running lines \ref{line:one}-\ref{line:endfor} of Algorithm \ref{alg:ig}, except 
with $T = (k+1)\log n +1$ iterations of the loop and finding minimum-cost $k$-path-cycle covers instead of 
the path-cycle covers on line \ref{line:findflow}. Then it finds $k$ $s$-$t$ paths in the resulting acyclic
 graph $F$, satisfying conditions of Lemma~\ref{lem:kpaths} below. The algorithm concludes by incorporating each 
component of the circulation $H$ into one of the obtained paths, similarly to lines \ref{line:cyc}-\ref{line:endcyc} 
of Algorithm \ref{alg:ig}. 

\begin{lemma} \label{lem:kpaths}
After lines \ref{line:one}-\ref{line:endfor} of Algorithm \ref{alg:ig} are executed with $T$ iterations of the 
loop on line \ref{line:iter} and finding minimum-cost $k$-path-cycle covers on line \ref{line:findflow}, there 
exist $k$ $s$-$t$ paths in the resulting acyclic graph $F$, such that each edge of $F$ is used at most $k$ times, 
and every node of $F$ is contained in at least one path. Moreover, these paths can be found in polynomial time.
\end{lemma}
\begin{proof}
We prove the existence result first. 
We note that the graph $F$ can support $kT$ units of flow from $s$ to $t$. This is because in each of the $T$ 
iterations, $k$ $s$-$t$ paths were added to the graph, whereas the removal of cycles does not decrease the 
amount of flow supported. So $F$ can be decomposed into $kT$ edge-disjoint paths from $s$ to $t$. Moreover,
 each node of $F$ participates in at least $T-\log n$ of these paths by Lemma \ref{lem:amtflow}.

Let $K$ be the comparability graph obtained from $F$. Namely, $K$ has the same set of nodes as $F$, and there
 is an undirected edge between nodes $u$ and $v$ in $K$ if $F$ contains a directed path either from $u$ to $v$
 or from $v$ to $u$. We claim that the nodes of $K$ can be partitioned into at most $k$ cliques. As comparability
 graphs are perfect graphs \cite{golumbic}, the minimum number of cliques that $K$ can be partitioned into is
 equal to the size of the maximum independent set in $K$. So suppose, for the sake of contradiction, that $K$ 
contains an independent set $I$ of size $k+1$. Then no nodes in $I$ must share any of the paths in our decomposition
 of $F$. Since each appears on at least $T-\log n$ paths, there must be at least $(k+1) \cdot (T-\log n)$ paths total 
in the decomposition. However, this is a contradiction, 
since $(k+1) \cdot (T-\log n) = (k+1)k \log n +k +1 > (k+1)k \log n +k = kT$.

Given the set of $k$ cliques in $K$, we can convert each of them into an $s$-$t$ path in $F$. For each
 such clique $C$, order the nodes of $C$ in a way consistent with a topological ordering of $F$. Then
 $F$ contains a path from each node $u$ of $C$ to the next node $v$: the existence of a $uv$ edge in $K$
 shows that there is a path between these nodes in $F$, and the topological ordering guarantees that this
 path goes in the correct direction. There must also be paths from $s$ to the first node of $C$ as well as
 from the last node of $C$ to $t$, since $s$ is the only source and $t$ is the only sink of $F$. Furthermore,
 the acyclicity of $F$ shows that each edge of $F$ is used at most once for connecting nodes in $C$. As there
 are $k$ cliques, each edge is used at most $k$ times in total. If the number of cliques is smaller than $k$,
 then we can add arbitrary paths from $F$ to our collection to obtain exactly $k$ paths. 

To find the required paths algorithmically, we construct the following bipartite graph $B$ from $F$ with bipartitions $X$ and $Y$. For each $v \in F \setminus \{s,t\}$, add nodes $x_v$ to $X$ and $y_v$ to $Y$ (note that $|X| = |Y|$). Now, for each ordered pair of nodes $(u,v)$ such that there is a directed path from $u$ to $v$ in $F$, add an edge from $x_u$ to $y_v$. Let $K'$ be the directed graph obtained from $K-\{s,t\}$ by orienting each edge $uv$ according to flow $F$. It is easy to see that any collection of $q$ disjoint cliques covering $K-\{s,t\}$ corresponds to a covering of $K'$ by $q$ vertex-disjoint paths and vice-versa. We claim that there is a covering of $K'$ by $q$ disjoint paths if and only if there is a matching in $B$ that leaves only $q$ nodes in $X$ and $q$ nodes in $Y$ unmatched.

Consider a covering of $K'$ by $q$ vertex-disjoint paths $P$. Form a matching $M = \{x_uy_v : uv \rm{~used~by~} P\}$.
Since $P$ is a collection of vertex-disjoint paths, each node has indegree and outdegree at most one in $P$, so
$M$ is a valid matching. Furthermore, since $P$ covers the nodes of $K'$ with $q$ paths, then exactly $q$ nodes
have indegree 0 and exactly $q$ nodes have outdegree 0. Thus, in $M$ exactly $q$ nodes of $X$ and $q$ nodes of $Y$ are unmatched.
Conversely, let $M$ be a matching of $B$ that leaves exactly $q$ nodes of $X$ unmatched. For each unmatched node $y_u \in Y$,
form a directed path in $K'$ starting at $u$ by the following process. If $x_u$ is matched to, say, $y_v$, then add arc $uv$
to the path. Continue this process from $y_v$ until the corresponding node in $X$ is not matched. This will
form $q$ vertex-disjoint paths in $K'$ that cover all nodes.

As noted earlier, we know that $K$ can be covered by at most $k$ disjoint cliques. Therefore, there is a matching in $B$ that
leaves at most $k$ nodes in $X$ and $k$ nodes in $Y$ unmatched. Compute such a matching and map it to the corresponding paths
in $K'$. Finally, extend these at most $k$ paths to $s-t$ paths by adding an arc from $s$ to the start of each path and an arc
from the end of each path to $t$. Again, since the only source in $F$ is $s$ and the only sink in $F$ is $t$, these arcs correspond
to edges in $K$.
\end{proof}

\begin{proofof}{Theorem \ref{thm:kperson}}
Let $L$ be the cost of a linear programming relaxation for the problem. 
The edges of $F$ as well as the edges used to connect the eulerian components of $H$ to the paths come from the union of $T$ $k$-path-cycle covers on subsets of $V$, and thus cost at most $T\cdot L = O(k\log n)\cdot L$. However, the algorithm may use each edge of $F$ up to $k$ times in the paths of Lemma \ref{lem:kpaths}, which makes the total cost of the produced solution at most $O(k^2 \log n)\cdot L$.
\end{proofof}


\section{Approximation algorithm for Directed Latency} \label{sec:latency}
We introduce LP relaxation (\ref{lp:lat}) for the directed latency problem.
The use of $x_{uw}$ and $f^v_{uw}$ variables in an integer programming formulation for directed latency was proposed by Mendez-Diaz et al.\ \cite{mendez-diaz}, who do a computational evaluation of the strength of its LP relaxation. However, our LP formulation uses a different set of constraints than the one in \cite{mendez-diaz}.

In LP (\ref{lp:lat}), a variable $x_{uw}$ indicates
that node $u$ appears before node $w$ on the path. Similarly, $x_{uvw}$ for three distinct nodes $u$, $v$, $w$ indicates that they appear in this order on the path. For every node $v\neq s$,
we send one unit of flow from $s$ to $v$, and we call it the $v$-flow. 
Then $f^v_{uw}$ is the amount of $v$-flow going through edge $(u,w)$, and $\ell(v)$ is the latency of node $v$.
To show that this LP is a relaxation of the directed latency problem, given a solution path $P$, 
we can set $f^v_{uw}=1$ whenever the edge $(u,w)$ is in $P$ and $v$ occurs later than $u$ in $P$, 
and $f^v_{uw}=0$ otherwise. So $f^v$ is one unit of flow from $s$ to $v$ along the path $P$.
Also setting the ordering variables $x_{uw}$ and $x_{uvw}$ to 0 or 1 appropriately and setting $\ell(v)$ to the latency of $v$ in $P$, we get a feasible solution to LP (\ref{lp:lat}) of the same cost as the total latency of $P$.

\begin{figure}[ht]
\begin{align}
\mbox{min}~ & \sum_{v\neq s} \ell(v) & \label{lp:lat}\\
\mbox{s.t.}~\, & \ell(v) \geq \sum_{uw} d_{uw} f^v_{uw}  
\hspace{33mm} \forall v \notag \\
& \ell(v) \geq [d_{su} + d_{uw} + d_{wv}]\, x_{uwv} \label{lk-const}  
\hspace{14mm} \forall u,w,v:|\{u,w,v\}|=3  \\
& \ell(t) \geq \ell(v) 
\hspace{45mm} \forall v \notag \\
& x_{uw}= x_{vuw} + x_{uvw} + x_{uwv} \label{xijk-const} 
\hspace{17mm} \forall u,w,v:|\{u,w,v\}|=3  \\
& x_{uw} + x_{wu} = 1  
\hspace{38mm} \forall u,w:u\neq w \label{xij-const} \\
& x_{su} = x_{ut} = 1  
\hspace{40mm} \forall u\notin \{s,t\} \notag \\
& \sum_w f^v_{wu} = \sum_w f^v_{uw}     
\hspace{33mm} \forall v, \forall u\notin \{s,v\} \label{vflcon} \\
& \sum_w f^v_{sw} = \sum_w f^v_{wv} = 1 
\hspace{27mm} \forall v \notag\\
& f^v_{us} = f^v_{vu} = 0  
\hspace{40mm} \forall u,v \label{fvus-const} \\
& \sum_w f^v_{uw} = x_{uv} 
\hspace{39mm} \forall v, u\neq v \label{fx-const} \\
& f^v_{uw} \leq f^t_{uw}   
\hspace{45mm} \forall u,w,v \label{univ-const}\\
& \sum_{u\notin S, w\in S} f^v_{uw} \geq x_{yv} 
\hspace{33mm} \forall S \subset V\setminus\{s\}, y\in S \label{lat-set-const} \\
& x_{uw}, x_{uwv}, f^v_{uw} \geq 0  
\hspace{32.5mm} \forall u,w,v   \notag
\end{align}
\end{figure}

Constraints (\ref{vflcon}) are the flow conservation constraints for $v$-flow at node $u$.
 Constraints (\ref{fvus-const}) ensure that no $v$-flow enters $s$ or leaves $v$. 
Constraints (\ref{fx-const}) say that $v$-flow passes through $u$ if and only if $u$ occurs before $v$. 
Since the $t$-flow goes through every vertex, when all the variables are in $\{0,1\}$, it defines an $s$-$t$ path.
We can think of the $t$-flow as the universal flow
and Constraints (\ref{univ-const}) ensure that every  $v$-flow follows an edge which
has a universal flow on it. 
Constraints (\ref{lat-set-const}), in an integer solution, ensure that if a set $S$ contains some node 
$y$ that comes before $v$ (i.e.\ $x_{yv}=1$), then at least one unit of $v$-flow enters $S$. 

We note 
that a min-cut subroutine can be used to detect violated constraints of type (\ref{lat-set-const}), allowing us to solve 
 LP (\ref{lp:lat}) using the ellipsoid method. 
Our analysis does not actually use 
Constraints (\ref{univ-const}), so we can drop them from the LP. Although without these
constraints the corresponding integer program may not be an exact formulation
for the directed latency problem,
we can still find a solution whose cost is within factor $O(\log n)$ of this relaxed LP.

\begin{lemma}\label{lem:scale}
Given a feasible solution to LP (\ref{lp:lat}) with objective value
$L$, we can find another solution of value at most $(1+\frac{1}{n})L$ in which the ratio of
the largest to smallest latency $\ell(\,)$ is at most $n^2$.
\end{lemma}
\begin{proof}
Let $({x},{\ell},{f})$ be a feasible solution with value $L$, with $\ell(t)$ the largest latency value in this solution. Note that $L\geq \ell(t)$.
Define a new feasible solution $({x},{\ell'},{f})$ by ${\ell'}(v)=\max\{\ell(v),\ell(t)/n^2\}$.
The total increase in the objective function is at most $n\cdot\frac{\ell(t)}{n^2}\leq L/n$ as there are $n$ nodes in total. Thus, the objective value of this new solution is at most $(1+1/n)L$.
\end{proof}

\medskip

Using Lemma \ref{lem:scale} and scaling the edge lengths (if needed), 
we can assume that we have a solution $({x},{\ell},{f})$ satisfying the following:

\begin{corollary} \label{cor:lstar}
There is a feasible solution $({x},{\ell},{f})$
in which the smallest latency is 1 and the largest latency is at most $n^2$ and whose cost is at most $(1+\frac{1}{n})$ times the optimum LP solution.
\end{corollary}

\noindent
Let $L^*$ be the value (i.e. total latency) of this solution.

The idea of our algorithm is to construct $s$-$v$ paths for several nodes $v$, such that together they cover all vertices of $V$, and then to ``stitch'' these paths together to obtain one Hamiltonian path. We use our results for ATSPP to construct these paths.  For this, we observe that parts of a solution to the latency  LP (\ref{lp:lat}) can be transformed to obtain feasible solutions to different instances of 
LP($\alpha$).  For example, we can construct a Hamiltonian $s$-$t$ path of total length $O(\log n)\cdot \ell(t)$ as follows. From a solution to LP (\ref{lp:lat}), take the $t$-flow defined by the variables $f^t_{uw}$, and notice that it constitutes a feasible solution to LP($\alpha=1$). In particular, since $x_{yt}=1$ for all $y$, constraints (\ref{lat-set-const}) of LP (\ref{lp:lat}) for  $v=t$ imply that the set constraints (\ref{alp-sc}) of LP (\ref{alphalp}) are satisfied.  The objective function value for LP (\ref{alphalp}) of this solution is at most $\ell(t)$. Thus, by Theorem \ref{thm:ig}, we can find the desired path. Of course, this path is not yet a good solution for the latency problem, as even nodes $v$ with $\ell(v) \ll \ell(t)$ can have latency in this path close to $O(\log n) \cdot\ell(t)$.  Our algorithm constructs several paths of different lengths, incorporating most nodes $v$ into paths of length $O(\log n) \cdot \ell(v)$, and then combines these paths to obtain the final solution.

\begin{algorithm*}[ht]
  \caption{~Directed Latency} \label{alg:latency} 
\begin{algorithmic}[1] 
\State Let $({x},{\ell},{f})$ be a solution to LP (\ref{lp:lat}). Let $S$ be the path $\{s\}$.
\State Partition the nodes into $g = \lfloor \log \ell(t)+1 \rfloor$
sets $V_1,\ldots,V_g$ with $v\in V_i$ if $2^{i-1}\leq\ell(v)< 2^i$.
\For {$i=1$ to $g-1$}\label{stp:forloop}
\For {$j=1$ to $2$}
\If {$V_i \neq \emptyset$}

\State Let $v_i^j = \argmax_{v\in V_i} |\{ u\in V_i : x_{uv} \geq \frac{1}{2} \}|$ \label{step:vij}
\Comment this maximizes the size of $B_i^j$ below

\State Let $A_i^j = \{u\in V: x_{uv_i^j} \geq \frac{2}{3}+\frac{2i-2+j}{24\log n}\}$
\State Let $B_i^j = \{u\in V_i: x_{uv_i^j} \geq \frac{1}{2} \}$
\Comment $|B_i^j| \geq (|V_i|-1)/2$

\State Find an $s$-$v^j_i$ path $P^j_i$, containing $A_i^j$, of cost $\delta_1 \log n \cdot 2^i$; append $P^j_i$ to $S$. \label{stp:pij}

\State Find $2 \log n$ $s$-$v^j_i$ paths ${\cal P}^j_i$, containing $B_i^j$, of total cost at most 
$2\log n\cdot 2^i$; append ${\cal P}^j_i$ to $S$. \label{stp:pijs}

\State $V_i = V_i \setminus (A_i^j \cup B_i^j \cup \{v_i^j\})$ \label{stp:decv}
\Comment size of $V_i$ is at least halved

\EndIf
\EndFor
\State Let $V_{i+1} = V_{i+1} \cup V_i$ \label{stp:move-nodes}
\Comment remaining nodes are carried over to the next set

\EndFor
\State Construct an $s$-$t$ path $P_g$, containing $V_g$, of cost at most $(2\log n+1)\cdot \ell(t)$.  Append $P_g$ to $S$. \label{stp:pst}
\State Shortcut $S$ over the later copies of repeated nodes. Output $S$.

\end{algorithmic}
\end{algorithm*}

\subsection{Constructing the paths}~
Algorithm \ref{alg:latency} finds an approximate solution to the directed latency problem, and we now explain how some of its steps are performed. The algorithm maintains a path $S$, initially containing only the source, and gradually adds new parts to it. This is done through operation \emph{append} on lines \ref{stp:pij}, \ref{stp:pijs}, and \ref{stp:pst}. To append a path $P$ to $S$ means to extend $S$ by connecting its last node to the first node of $P$ that does not already appear in $S$, and then following until the end of $P$. For example, if $S=sabc$ and $P=sbdce$, the result is $S=sabcdce$. Step \ref{stp:pijs} appends a set of paths to $S$. This just means sequentially appending all paths in the set, in arbitrary order, to $S$.

Next we describe how to build paths $P^j_i$ and ${\cal P}^j_i$ in Steps \ref{stp:pij} and \ref{stp:pijs}. We described above how to use Theorem \ref{thm:ig} to build a Hamiltonian $s$-$t$ path $P$ of length $(2 \log n+1)\cdot \ell(t)$, which is used on line \ref{stp:pst} of the algorithm.  The idea behind building paths $P^j_i$ and ${\cal P}^j_i$ with their corresponding length guarantees is similar.

To construct $P^j_i$, we do the following.
Since each node $u\in A^j_i$ has $x_{uv^j_i} \geq 2/3$,  the amount of ${v^j_i}$-flow
that goes through $u$ is at least $2/3$. We apply splitting-off on this flow to nodes outside of  $A^j_i$, and obtain a total of one unit of $s$-$v^j_i$ flow over the nodes in $A^j_i$, of cost no larger than $\ell(v^j_i)\leq 2^i$. This flow satisfies all the constraints of LP($\alpha=2/3$), including the set constraints (\ref{alp-sc}), which are implied by the set constraints (\ref{lat-set-const}) of the latency LP (\ref{lp:lat}), as $x_{uv^j_i}\geq 2/3$ for $u\in A^j_i$.
Thus, using Theorem \ref{thm:a-ig}, we can find a path from $s$ to $v^j_i$, spanning
all the nodes of $A^j_i$, whose cost is at most $\delta_1 \log n \cdot 2^i$ for some constant $\delta_1$.

To obtain the set of paths ${\cal P}^j_i$, we look at the $v^j_i$-flow going through each node of $B^j_i$, whose amount is at least $\frac{1}{2}$. After splitting-off all nodes outside of $B^j_i$, we get a feasible solution of cost at most $\ell(v^j_i)\leq 2^i$ to LP($\alpha = 1/2$). 
By Theorem \ref{thm:kpath}, we can find 
$2 \log n$ paths, each going from $s$ to $v^j_i$, which together cover all the nodes of $B^j_i$, and whose total cost is at most $2 \log n\cdot 2^i$.

\subsection{Connecting the paths}~
We now bound the lengths of edges introduced by the append operation in the different cases. For a path $P$,
let $app(P)$ be the length of the edge used for appending $P$ to the path $S$ in the algorithm.

\begin{lemma} \label{lem:app-p}
For any $i$, $j$, and path $P\in {\cal P}^j_i$, $app(P)\leq 6\cdot 2^i$. Also, $app(P_g) \leq 6 \cdot 2^g$.
\end{lemma}
\begin{proof}
Let $u$ be the last node of the path $S$ before the append operation, $v^j_i$ be the last node of $P$, and $w$ be the first node of $P$ that does not appear in $S$. We need to bound $d_{uw}$, the distance from $u$ to $w$. 

We observe that $x_{wu} \leq 5/6$. If $u=s$, this is trivial. Otherwise, $u=v^{j'}_{i'}$ is the endpoint of some path constructed in an earlier iteration. Note that  $j'\leq 2$ and  $i'\leq g-1 \leq \log \ell(t) \leq 2 \log n$ by our assumption that $\ell(t)\leq n^2$, which means that $\frac{5}{6} \geq \frac{2}{3}+\frac{2i'-2+j'}{24\log n}$. So, if we had $x_{wu} > 5/6$, then $w$ would be included in the set $A^{j'}_{i'}$ and in the path $P^{j'}_{i'}$, and thus be already contained in $S$, which is a contradiction.

Consequently, $x_{uw} = 1 - x_{wu} \geq 1/6$. 
This means that the amount of $w$-flow that goes through $u$ is at least $1/6$. Since this flow has to reach $w$ after visiting $u$, it has to cover a distance of at least $d_{uw}$, thus adding at least $\frac{1}{6}\cdot d_{uw}$ to $\ell(w)$, the latency of $w$. Thus, $\ell(w)\geq \frac{1}{6} d_{uw}$, and $d_{uw} \leq 6 \ell(w)$. Now, if $w\in {\cal P}^j_i$, it must be in $B^j_i$, which, by definition, means that $w\in V_i$, and therefore $\ell(w)\leq 2^i$. So $app(P) = d_{uw}\leq 6\cdot 2^i$. If $w\in P_g$, then $app(P_g) \leq 6 \ell(w) \leq 6 \ell(t) \leq 6\cdot 2^g$.
\end{proof}

\medskip

To bound the cost of appending a path $P^j_i$ to $S$, we need an auxiliary lemma.

\begin{lemma} \label{lem:ijk}
For any $\epsilon > 0$, if $x_{uw}+x_{wv} \geq 1+\epsilon$, then $\ell(v) \geq \epsilon \cdot d_{uw}$.
\end{lemma}
\begin{proof}
Using Constraint (\ref{xijk-const}) we have:
\begin{eqnarray*}
1+\epsilon &\leq& x_{uw} + x_{wv}   \\
&=& (x_{vuw} + x_{uvw} + x_{uwv}) 
  + (x_{uwv} + x_{wuv} + x_{wvu}) \\
 & = &  2 x_{uwv} + (x_{vuw} + x_{uvw})  + (x_{wuv} + x_{wvu}).
\end{eqnarray*}

On the other hand, $(x_{vuw} + x_{uvw}) + (x_{wuv} + x_{wvu}) \leq x_{vw} + x_{wu} = 2 - (x_{uw} + x_{wv}) \leq 1-\epsilon$, 
using again Constraint (\ref{xijk-const}), then Constraint (\ref{xij-const}), and the assumption of the lemma.
Therefore,  $2 x_{uwv} \geq (1+\epsilon) - (1-\epsilon) = 2\epsilon$, i.e. $x_{uwv} \geq \epsilon$.  
Then the claim follows using Constraint (\ref{lk-const}).
\end{proof}

\begin{lemma} \label{lem:app-pij}
For any $i$ and $j$, $app(P^j_i)\leq 24 \log n \cdot 2^i$.
\end{lemma}
\begin{proof}
Let $u$, $v^j_i$, and $w$ be as in the proof of Lemma \ref{lem:app-p}.
To bound $d_{uw}$, we consider two cases.

{\bf Case 1:} If $w\in V_i$, we apply the same proof as for Lemma \ref{lem:app-p} and conclude that $app(P^j_i)\leq 6 \cdot 2^i$.

{\bf Case 2:} If $w\not\in V_i$, let $(i',j')$ be an earlier iteration of the algorithm in which node $u = v^{j'}_{i'}$ was added to $S$. Since $w\notin S$, it must be that $w\notin A^{j'}_{i'}$, and thus $x_{wu} < \frac{2}{3}+\frac{2i'-2+j'}{24\log n}$. On the other hand, since $w\in A^j_i$, it must be that $x_{wv^j_i} \geq \frac{2}{3}+\frac{2i-2+j}{24\log n}$. Because $2i'+j' \leq 2i+j-1$, we have
\begin{eqnarray*}
x_{uw}+x_{wv^j_i}  &=& (1 - x_{wu})+x_{wv^j_i} \\
 &\geq & 1 - \frac{2i'-2+j'}{24\log n} +\frac{2i-2+j}{24\log n} \\
 &\geq & 1 + \frac{1}{24\log n}.
\end{eqnarray*}
Using Lemma \ref{lem:ijk}, we get that 
$app(P^j_i) = d_{uw} \leq 24 \log n \cdot \ell(v^j_i) \leq 24 \log n \cdot 2^i$. 
\end{proof}

\begin{lemma}\label{lem:Pi-late}
Suppose that a node $v$ is first added to path $S$ in iteration $k$ of the outer loop of the algorithm. Then the latency of $v$ in $S$ is at most $\delta_2 \log n \cdot 2^k$, for some constant $\delta_2>0$.
\end{lemma}
\begin{proof}
Let $len(P)$ denote the length of a path $P$.  The latency of node $v$ on $S$ is at most:
\begin{eqnarray*}
&& \sum_{i=1}^k \sum_{j=1}^2 \left[ 
len(P^j_i) + \sum_{P\in {\cal P}^j_i} len(P) 
+ app(P^j_i) + \sum_{P\in {\cal P}^j_i} app(P) \right] \\
&\leq& 
\sum_{i=1}^k \sum_{j=1}^2 \left[ 
\delta_1 \log n \cdot 2^i + 2 \log n \cdot 2^i   
+ 24 \log n \cdot 2^i + 2 \log n \cdot 6 \cdot 2^i \right] \\
 & \leq & \delta_2 \log n \cdot 2^k  
\end{eqnarray*}
\end{proof}

\medskip

Suppose that $n_i$ is the number of nodes that are originally placed into the set  $V_i$. Since a node $v$ is originally placed in $V_i$ if $\ell(v)\geq 2^{i-1}$, the value of the LP solution $L^*$ can be bounded by:

\begin{equation}\label{eqn:opt-lower}
L^* ~=~ \sum_v \ell(v)  ~\geq~ \sum_{i=1}^g n_i \, 2^{i-1}.
\end{equation}

Let $n'_i$ denote the size of $V_i$ at the beginning of iteration $i$ of the outer loop. Note that $n'_i$ may be larger than $n_i$ since some nodes may have been moved to $V_i$ in Step \ref{stp:move-nodes} of the previous iteration.

\begin{claim} \label{cl:cov}
For any $i$, the size of the set $V_i$ at the end of iteration $i$ is at most $n'_i/4$.
\end{claim}
\begin{proof}
Consider the iteration $(i,j=1)$.  Note that the vertex $v^j_i$ is chosen precisely to maximize the number of nodes $u$ in $V_i$ with $x_{uv^j_i}\geq 1/2$, which is the size of the set $B^j_i$. If we imagine a directed graph $H$ on the set of vertices $V_i$, in which an edge $(u,w)$ exists whenever $x_{uw}\geq 1/2$, then $v^j_i$ is the vertex with highest in-degree in this graph. Now, from  Constraint (\ref{xij-const}), it's not hard to see that some vertex in $H$ will have in-degree at least $(n'-1)/2$. So the number of nodes removed from $V_i$ in step \ref{stp:decv} of the algorithm is at least $|B^j_i \cup \{v^j_i\}| \geq n'/2$, and size of $V_i$ decreases at least by a factor of two. Similarly, at least half of the remaining nodes of $V_i$ are removed in the iteration $j=2$, so overall the size of $V_i$ decreases at least by a factor of four.
\end{proof}

\medskip

We now show that the total latency of the final solution $S$ 
is at most $O(\log n)\cdot L^*$.

\begin{proofof}{Theorem \ref{thm:lat}}
From Claim \ref{cl:cov}, it follows that at most
a $1/4$ fraction of the $n'_i$ nodes that are in $V_i$ at the beginning of iteration $i$ are moved to the set $V_{i+1}$ at the end of this iteration. Thus, for any $1<i\leq g$,~ $n'_i \leq n_i + {n'_{i-1}}/{4}$. This implies that $n'_i \leq \sum_{h=1}^i n_h/4^{i-h}$.

Now we claim that the total latency of the solution $S$ is at most
$\sum_{i=1}^{g} n'_i \cdot \delta_2 \log n \cdot 2^i$.
This is because at most $n'_i$ nodes are added to $S$ in iteration $i$, and each such node has latency
at most $\delta_2 \log n \cdot 2^i$ (using Lemma \ref{lem:Pi-late}). 
Therefore, the total latency of the solution is at most:
\begin{eqnarray*}
\sum_{i=1}^{g} n'_i \cdot \delta_2 \log n \cdot 2^i
&\leq& \sum_{i=1}^{g} \delta_2 \log n \cdot 2^i \cdot \sum_{h=1}^i \frac{n_h}{4^{i-h}} \\
&=& \delta_2 \log n \sum_{i=1}^{g} \sum_{h=1}^i 2^{h-i} \cdot  2^h \, n_h \\
&\leq& \delta_2 \log n \sum_{h=1}^{g} 2^h\, n_h \sum_{i=0}^\infty \frac{1}{2^i} \\
&\leq& O(\log n) \cdot L^*,
\end{eqnarray*}
using the bound on $n'_i$, re-ordering the summation, and using inequality (\ref{eqn:opt-lower}). Combined with Corollary \ref{cor:lstar}, this proves the theorem.
\end{proofof}


\subsection{Extensions}
The above algorithm can be easily extended to the more general setting in which every node of the graph comes with a weight $c(v)$ and the goal is to find a Hamiltonian $s$-$t$ path to minimize the total {\em weighted} latency, where the weighted latency of a node $v_i$ is equal to 
$c(v)\cdot\sum_{j=1}^{i-1} d_{v_jv_{j+1}}$. 
This requires changing the objective function of LP (\ref{lp:lat}) to $\sum_{v\neq s} c(v)\ell(v)$ and changing the definition of $v_i^j$ on line \ref{step:vij} of Algorithm \ref{alg:latency} to maximize the total weight, instead of the number, of vertices in $B_i^j$.

We also note that our approximation guarantee for directed latency is of the form $O(\gamma + \log n)$, where $\gamma$ is the integrality gap of ATSPP. So an improvement of the bound on $\gamma$ would not immediately lead to an improvement for directed latency. 

\section*{Acknowledgements} This work was partly done while the second author was visiting Microsoft Research New England; he thanks MSR for hosting him. We also thank the anonymous referees for helpful comments.

\bibliographystyle{plain}
\bibliography{../../bib/names,../../bib/conferences,../../bib/bibliography}

\end{document}